\documentclass[11pt,a4paper]{amsart}
\usepackage{amsmath,amssymb,mathrsfs,amsthm}

\newtheorem{thm}{Theorem}[section]

\newtheorem{lemma}{Lemma}[section]

\theoremstyle{definition}
\newtheorem{defi}{Definition}

\newtheorem{assumption}{Assumption}

\theoremstyle{remark}
\newtheorem{remark}{Remark}[section]

\numberwithin{equation}{section}

\title[Resonances for Wigner-von Neumann type Hamiltonian]{Resonances and viscosity limit for \\
the Wigner-von Neumann type Hamiltonian}

\author{Kentaro Kameoka}
\address{Graduate School of Mathematical Sciences, University of Tokyo,
 3-8-1, Komaba, Meguro-ku, Tokyo 153-8914, Japan.}
\email{kameoka@ms.u-tokyo.ac.jp}
\thanks{KK is partially supported by FMSP program at the 
Graduate School of Mathematical Sciences, the University of Tokyo.}

\author{Shu Nakamura}
\address{Department of Mathematics, Faculty of Sciences, Gakushuin University
1-5-1, Mejiro, Toshima, Tokyo 171-8588, Japan.}
\email{shu.nakamura@gakushuin.ac.jp}
\thanks{SN thanks Professor J.-F. Bony for suggesting to consider the problem using methods 
of \cite{N2}. He also thanks Professor M. Zworski for valuable discussions and encouragements.
SN is partially supported by JSPS Grant Kiban (B) 15H03622. } 
\date{\today}
\begin{document}

\maketitle

\begin{abstract}
The resonances for the Wigner-von Neumann type Hamiltonian are defined by the periodic 
complex distortion in the Fourier space. 
Also, following Zworski, we characterize resonances as the limit points of discrete eigenvalues 
of the Hamiltonian with a quadratic complex absorbing potential in the viscosity type limit.
\end{abstract}

\section{Introduction}
In this paper, we consider the one dimensional Schr\"odinger operator
\[
P=-\frac{d^2}{dx^2}+V(x) \hspace{0.3cm} \text{on} \hspace{0.3cm}L^2(\mathbb{R})
\]
and its resonances, where $V(x)$ is a oscillatory and decaying potential. We assume 

\begin{assumption}\label{ass-1} 
The potential $V(x)$ has the following form: 
\[
V(x)=\sum_{j=1}^{J}s_j(x)W_j(x),
\]
where $J\in \mathbb{N}$, 
$s_j \in C(\mathbb{R}; \mathbb{R})$ are periodic functions with period $\pi$ whose Fourier serieses converge 
absolutely, and
$W_j \in C^{\infty}(\mathbb{R}; \mathbb{R})$ have analytic continuations to the region 
$\{z=x+iy\mid |x|>R_0, |y|<K|x|\}$ for some $R_0>0$ and $K>0$ with the bound $|W_j(z)|\le C|z|^{-\mu}$ 
for some $\mu>0$ in this region.
\end{assumption}

A typical example is $V(x)=a\frac{\sin 2x}{x}$, where $a \in \mathbb{R}$. 
We note that $P$ is not dilation analytic in this case since the potential is exponentially growing in the complex direction. 
We also note that dilation analytic potentials satisfy Assumption~\ref{ass-1} by setting $s_j (x)=1$. 
We first show that resonances can be defined for this class of potentials. 
We write the set of threshold by $\mathcal{T}=\{n^2\mid n\in \mathbb{N}\cup\{0\}\}$. 
The resolvent on the upper half plane is denoted by $R_+(z)=(z-P)^{-1}$, $\mathrm{Im}z>0$. 

\begin{thm}\label{thm-1}
Under Assumption~\ref{ass-1}, there exists a complex neighborhood:  
$\Omega\subset \mathbb{C}$ of $[0, \infty)\setminus \mathcal{T}$
such that the following holds: 
For any $f, g \in L_{\mathrm{comp}}^{2}(\mathbb{R})$, 
the matrix element $(f, R_+(z)g)$ has a meromorphic continuation to $\Omega$. 
\end{thm}

\begin{remark}
$\Omega$ in Theorem~\ref{thm-1} and Theorem~\ref{thm-2}  are given explicitly in Section 2 and Section 3 respectively.
\end{remark} 
\begin{remark}
Unfortunately, the original Wigner-von Neumann potential (\cite{NW}, see also \cite[Section~XIII.13]{RS}) 
\[V(x)=(1+g(x)^2)^{-2}(-32\sin x)(g(x)^3\cos x-3g(x)^2 \sin^3 x+g(x)\cos x +\sin^3 x),\]
where $g(x)=2x-\sin 2x$, does not seem to satisfy Assumption~\ref{ass-1}. 
In fact, the argument principle implies that if $\nu>1/2$ and $\ell\gg 1$ with $\ell\in \mathbb{Z}$, 
$g(z)\pm i$ have two zeros in the region 
$\{z \in \mathbb{C} \mid (\ell-1/2)\pi \le \mathrm{Re}z\le (\ell+1/2)\pi,  
-\nu \log \ell\le \mathrm{Im}z\le \nu \log \ell\}$. 
Thus another method is needed to study the complex resonances for the original 
Wigner-von Neumann Hamiltonian. 
\end{remark}

Following the standard theory of resonances, the complex resonances are defined using this meromorphic continuation. 

\begin{defi}
Let $R_+(z)$ be the meromorphic continuation of the resolvent for $P$ as in 
Theorem~\ref{thm-1}. A complex number $z\in \Omega$ is called a resonance if $z$ is a pole of $(f, R_+(z)g)$ 
for some $f, g \in L_{\mathrm{comp}}^{2}(\mathbb{R})$
and the multiplicity $m_z$ is defined as the maximal number $m$ such that 
there exist $ f_1, \dots ,f_m , g_1, \dots, g_m \in L_{\mathrm{comp}}^{2}(\mathbb{R})$ 
with  $\det\bigl(\frac{1}{2\pi i}\oint_{C(z)}(f_i, R_{+}(\zeta)g_j)d\zeta\bigr)_{i,j=1}^{m} \not =0 $, 
where $C(z)$ is a small circle around $z$. 
The set of resonances is denoted by $\mathrm{Res}(P)$. 
\end{defi}
\begin{remark}
$\mathrm{Res}(P)$ is discrete and $m_z<\infty$ for any $z \in \Omega$ (see Remark~\ref{rem-2-2}).
\end{remark}

We now introduce a complex dissipative potential as follows: 
\[
P_{\varepsilon}=-\frac{d^2}{dx^2}+V(x)-i\varepsilon x^2,\quad \varepsilon >0. 
\]
We easily see that $P_{\varepsilon}$, $\varepsilon>0$, has purely discrete spectrum on $L^2(\mathbb{R})$. 
Zworski~\cite{Z} proved that the set of resonances can be characterized as limit points of 
the eigenvalues of $P_\varepsilon$ as $\varepsilon \to 0$, 
namely $\lim_{\varepsilon \to 0}\sigma_d(P_{\varepsilon})=\mathrm{Res}(P)$  for the dilation analytic case. 
Zworski~\cite{Z} also proposed a problem of finding a potential $V(x)$ such that the limit set of 
$\sigma_d(P_{\varepsilon})$ when $\varepsilon \to 0$ 
is not discrete, and suggested $V(x)= \frac{\sin x}{x}$ as a candidate for such $V(x)$. 
Our next result disproves this conjecture. 

\begin{thm}\label{thm-2}
Under Assumption~\ref{ass-1}, there exists a complex neighborhood $\Omega\subset \mathbb{C}$ 
of $[0, \infty)\setminus \mathcal{T}$ such that 
$\lim_{\varepsilon \to 0}\sigma_d(P_{\varepsilon})=\mathrm{Res}(P)$ in $\Omega$ including multiplicities. 
In particular, $\lim_{\varepsilon \to 0}\sigma_d(P_{\varepsilon})$ is discrete in $\Omega$.
More precisely, for any $z \in \Omega$ there exists $\rho_0>0$ such that 
for any $0<\rho<\rho_0$ there exists $\varepsilon_0>0$ 
such that for any $0<\varepsilon<\varepsilon_0$, 
 \[ \# \sigma_d(P_{\varepsilon})\cap B(z, \rho)=m_z,\]
 where $B(z, \rho)=\{w \in \mathbb{C}\mid |w-z|\le \rho\}$.
\end{thm}

The main idea of Theorem~\ref{thm-1} is as follows: 
We note the standard dilation analytic method for the complex resonances does not apply to our 
potentials. On the other hand, it is known that if we set 
\[
A'=\frac{1}{2}(x\cdot D'+D'\cdot x), \hspace{0.3cm} D'u(x)=\frac{1}{2\pi}(u(x+\pi)-u(x-\pi)),
\]
then we can construct a Mourre theory with this conjugate operator (see \cite{N2}). 
We may use this operator as the generator of complex distortion to define the resonances 
for our model. 
Actually, in the Fourier space, $A'$ is a differential operator 
\[
\widetilde{A'}=(i\partial_{\xi})\cdot \sin(\pi \xi)+\sin(\pi \xi)\cdot (i \partial_{\xi}), 
\]
and this generates a periodic complex distortion in the Fourier space 
(see \cite{N} for Hunziker-type local distortion in the Fourier space). 

As in \cite{Z}, the essential ingredient of the proof of Theorem~\ref{thm-2} is the resolvent estimate of 
the distorted operator which is uniform with respect to $\varepsilon$ in the case of $V=0$. 
We prove this by employing the semiclassical analysis in the Fourier space with 
the semiclassical parameter $h=\sqrt{\varepsilon}$. 
Since we work in the Fourier space, the term $-i\varepsilon x^2=i\varepsilon \partial_{\xi}^2$ 
is the usual viscosity term (multiplied by $i$) and the viscosity limit corresponds to 
the semiclassical limit.  
%We also note that, while the potential is assumed to be compactly supported in \cite{Z}, 
%we remove this assumption by approximating the potential by compactly supported functions.

Wigner-von Neumann type Hamiltonians have been investigated by many authors. 
See for instance \cite{B1}, \cite{B2}, \cite{CHM}, \cite{DMR}, \cite{FH}, \cite{HKS}, \cite{K}, \cite{L}, 
\cite{R}, \cite{RUU} and references therein. 
To our knowledge, the definition of the complex resonances based on the complex distortion 
for Schr\"{o}dinger operators with oscillatory and slowly decaying potentials is new. 
The complex distortion in the momentum variables is studied by Cycon~\cite{C} and Sigal~\cite{S} for 
radially symmetric dilation analytic, or sufficiently smooth exponentially decaying potentials. 
In \cite{N}, this method is extended to not necessarily radially symmetric case. 
See the references in \cite{N} for related earlier works on the complex distortion.

Stefanov~\cite{St} studied the approximation of resonances by the fixed complex absorbing potential method 
in the semclassical limit. 
Similar methods are used in generalized geometric settings in Nonnenmacher-Zworski~\cite{NZ1}, \cite{NZ2} 
and Vasy~\cite{V}. 
As mentioned above, Theorem~\ref{thm-2} was proved by Zworski~\cite{Z} in the case of the 
dilation analytic case. 
Analogous results were proved for Pollicott-Ruelle resonances by Dyatlov-Zworski~\cite{DZ} 
(see also \cite{DR}, \cite{Dr}), 
and for 0th order pseudodifferential operators by Galkowski-Zworski~\cite{GZ}.  
For the numerical results and original approach in physical chemistry, 
see the references in \cite{St}, \cite{Z}.

This paper is organized as follows. 
In Section 2, we introduce the periodic complex distortion in the Fourier space and prove Theorem~\ref{thm-1}.
In Section 3, we prove Theorem~\ref{thm-2} employing the complex distortion introduced in Section 2.

%%%
\section{Periodic distortion in the Fourier space}
We introduce the following periodic distortion in the Fourier space
\[\Phi_{\theta}(\xi)=\xi+\theta \sin(\pi \xi), \hspace{0.3cm} 
U_{\theta}f(\xi)=\Phi_{\theta}'(\xi)^{\frac{1}{2}}f(\Phi_{\theta}(\xi)),\]
where $\theta \in (-\pi^{-1}, \pi^{-1})$.
In the Fourier space, $P$ has the form $\widetilde{P}=\xi^2+\widetilde{V}$, 
where $\widetilde{V}=(2\pi)^{-1/2}\hat{V} *$ is a convolution operator and 
$\hat{V}$ is the Fourier transform $\hat{V}(\xi)=(2\pi)^{-1/2}\int V(x)e^{-ix\xi}dx$.  
Hence we have 
\[
\widetilde{P}_{\theta}:= U_{\theta}\widetilde{P} U_{\theta}^{-1}
=(\xi+\theta\sin(\pi \xi))^2+\widetilde{V}_{\theta}, 
\quad \widetilde{V}_{\theta}=U_{\theta}\widetilde{V}U_{\theta}^{-1}. 
\] 

\subsection{Analyticity of $\widetilde{V}_{\theta}$}
\begin{lemma}\label{lem-1}
Under Assumption~\ref{ass-1}, $\widetilde{V}_{\theta}$ is analytic with respect to $\theta$ and 
$\xi^2$-compact for $\theta$ in some complex neighborhood of 
$\{i\delta \mid -K<\delta<K\}$, where $K$ is the constant in Assumption~\ref{ass-1}.
\end{lemma}

\begin{remark}\label{rem-2-1}
In the case of $V(x)=a\frac{\sin 2x}{x}$, we have $\widetilde{V}=\frac{a}{2}\chi_{[-2, 2]}*$.  
Then by  a simple computation we learn 
$ \widetilde{V}_{\theta}=(\Phi_{\theta}')^{\frac{1}{2}}\widetilde{V}(\Phi_{\theta}')^{\frac{1}{2}}$, 
where $(\Phi_{\theta}')^{\frac{1}{2}}$ is a multiplication operator by $\Phi_{\theta}'(\xi)^{\frac{1}{2}}$.
Thus the Lemma 1 is simpler to show in this case. 
\end{remark}
\begin{proof}[Proof of Lemma~\ref{lem-1}]
For real $\theta$, the integral kernel $\widetilde{V}_{\theta}(\xi, \eta)$ of $\widetilde{V}_{\theta}$ is given by
\[
\widetilde{V}_{\theta}(\xi, \eta)=
\Phi_{\theta}'(\xi)^{\frac{1}{2}}\hat{V}(\Phi_{\theta}(\xi)-\Phi_{\theta}(\eta))\Phi_{\theta}'(\eta)^{\frac{1}{2}}, 
\quad \xi,\eta\in\mathbb{R}. 
\]
We first consider the case of $V \in C^{\infty}_c(\mathbb{R}; \mathbb{R})$. 
Then the Paley-Wiener estimate implies that $\widetilde{V}_{\theta}(\xi, \eta)$ is analytic with respect to $\theta \in \mathbb{C}$ 
 and has the off-diagonal decay bounds 
\[ 
|\partial_{\xi}^{\alpha}\partial_{\eta}^{\beta} \widetilde{V}_{\theta}(\xi, \eta)|\le C_{\alpha, \beta, N}\langle \xi-\eta \rangle^{-N}, \quad \xi,\eta\in\mathbb{R}
\]
for any $\alpha, \beta$ and $N$, 
where $C_{\alpha, \beta, N}$ is independent of $\theta$ when $\theta\in \mathbb{C}$ ranges over a bounded set.
This and the formula 
\[\widetilde{V}_{\theta}=a^{\mathrm{w}}(\xi, D_{\xi}; \theta), \hspace{0.3cm} 
a(\xi, x; \theta)=\int \widetilde{V}_{\theta}\Bigl(\xi+\frac{\eta}{2}, \xi-\frac{\eta}{2}\Bigr)e^{-i\langle \eta, x\rangle}d\eta,\]
where $a^{\mathrm{w}}$ denotes the Weyl quantization (see, e.g., \cite[Section~8.1]{Z2}), 
imply that $\widetilde{V}_{\theta}$ is a pseudodifferential operator in the Fourier space with a symbol rapidly
decaying with respect to $x$ and analytic with respect to $\theta$. 
Thus the Lemma~\ref{lem-1} is proved in this case.

We next consider the case of $V(x)=s(x)W(x)$, where $s(x)$ and $W(x)$ satisfy the condition in Assumption~\ref{ass-1}. 
We first estimate the Fourier transform of $W(x)$. 
By the deformation of the integral, we have 
\[
\hat{W}(\xi)=(2\pi)^{-1/2}\int_{C_\pm, \tau}  W(z)e^{-iz\xi}dz,\quad \pm \xi>0,
\] 
where 
\[
C_{\pm, \tau}=(e^{\pm i\tau}(-\infty, 0]-2R_0)\cup [-2R_0, 2R_0]\cup (2R_0+e^{\mp i\tau}[0, \infty)),
\] 
$0<\tau<\arctan K$, and $R_0$ is that in Assumption~\ref{ass-1}.
This expression shows that $\hat{W}(\xi)$ has an analytic continuation to 
\[
S_{\tau}=\{z\in \mathbb{C}^*|
-\tau<\arg z<\tau\}\cup \{z\in \mathbb{C}^*|-\tau<\arg z-\pi<\tau\}.
\] 
$\hat{W}(\xi)$ decays rapidly in $S_{\tau}$ when $|\xi|\to \infty$ thanks to the smoothness of $W$. 
For small $\xi\in S_{\tau}$, we have $|\hat{W}(\xi)|\le C|\xi|^{-\frac{1}{1+\mu}}$, 
where $\mu>0$ is that in Assumption~\ref{ass-1}. 
To see this, we take $C_{\pm, \tau'}$ for $0<\tau<\tau'<\arctan K$ and estimate 
\[|\hat{W}(\xi)|\le C\int_0^{\infty}e^{-cx|\xi|}\langle x \rangle^{-\mu}dx 
= C|\xi|^{-1}\int_0^{\infty}e^{-c|x|}\langle x/|\xi| \rangle^{-\mu} dx.\]
We divide the integral into $\int_0^{\varepsilon}+\int_{\varepsilon}^{\infty}$ and we obtain the bound 
$ \frac{\varepsilon}{|\xi|}+\frac{1}{|\xi|}\langle \varepsilon /|\xi| \rangle^{-\mu}$. 
Taking $\varepsilon=|\xi|^{\frac{\mu}{1+\mu}}$, we have $|\hat{W}(\xi)|\le C|\xi|^{-\frac{1}{1+\mu}}$. 
We denote the Fourier transform of $s$ by $\hat{s}(\xi)=\sqrt{2\pi}\sum_{k\in \mathbb{Z}}a_k \delta (\xi-2k)$.
By Assumption~\ref{ass-1}, we have 
$\sum_{k \in \mathbb{Z}}|a_k|<\infty$. 
This and the estimates on $\hat{W}(\xi)$ above show that the Fourier transform 
\[
\hat{V}(\xi)=\sum_{k\in \mathbb{Z}} a_k \hat{W}(\xi-2k)
\]
has an analytic continuation to the region 
$T_{\tau}=\bigcup_{k\in \mathbb{Z}}T_{\tau, k}$, where 
\[
T_{\tau, k}=\{z\in \mathbb{C}\setminus\{0, 2\}|-\tau<\arg z<\tau, 
-\tau<\arg (2-z)<\tau\}+2k.
\]
Moreover, we have 
\begin{equation}\label{eq-2-1}
\sum_{k \in \mathbb{Z}}\sup_{T_{\tau, k}}|\xi-2k|^{\frac{1}{1+\mu}}|\xi-2k-2|^{\frac{1}{1+\mu}}|\hat{V}(\xi)|<\infty
\end{equation}
By \eqref{eq-2-1}, we have $|\widetilde{V}_{\theta}(\xi, \eta)|\le g(\xi-\eta)$ 
for some integrable function $g$. 
This is also true for $\frac{d}{d\theta} \widetilde{V}_{\theta}(\xi, \eta)$ by Cauchy's formula with respect to $\theta$. 
Thus Young's inequality implies that the operator $\widetilde{V}_{\theta}$ with integral kernel 
$\widetilde{V}_{\theta}(\xi, \eta)$ is $L^2$-bounded and analytic with respect to $\theta$.  
We note that if $\theta$ is purely imaginary, we have
\[
|\mathrm{Im}(\Phi_{\theta}(\xi)-\Phi_{\theta}(\eta))|
\le |\theta||\mathrm{Re}(\Phi_{\theta}(\xi)-\Phi_{\theta}(\eta))-2k|.
\]
Thus $\theta$ may be taken from a complex neighborhood of $\{i\delta | -\tan \tau <\delta<\tan \tau\}$. 
Since $0<\tau<\arctan K$ is arbitrary, $\widetilde V_\theta$ is analytic for $\theta$  as claimed in Lemma~\ref{lem-1}. 

To see $\xi^2$-compactness, we approximate $V$ by $C_c^{\infty}$ functions. 
Take $\chi \in C_c^{\infty}(\mathbb{R})$ such that $\chi=1$ near $x=0$. 
We decompose $V(x)=V_{1, R}+V_{2, R}$, where $R>0$,  
\[
V_{1, R}=\chi(x/R)W(x)\sum_{|k|\le R}a_k e^{2ikx}
\]
and 
\[
V_{2, R}=W(x)\sum_{|k|> R}a_k e^{2ikx}+(1-\chi(x/R))W(x)\sum_{|k|\le R}a_k e^{2ikx}. 
\]
We also denote the corresponding distorted operator on the Fourier space by
$\widetilde{V}_{\theta, 1, R}$ and $\widetilde{V}_{\theta, 2, R}$.  
Since $V_{1, R}\in C_c^{\infty}$, $\widetilde{V}_{\theta, 1, R}$ is $\xi^2$-compact. 
We also see that $\lim_{R\to \infty}\|\widetilde{V}_{\theta, 2, R}\|_{L^2 \to L^2}=0$ 
by the estimate for $V=s(x)W(x)$ as above. 
These completes the proof of Lemma~\ref{lem-1}
\end{proof}

\subsection{Definition of resonances}
By Lemma~\ref{lem-1} we learn that $\widetilde{P}_{\theta}$ is analytic with respect to $\theta$ 
in the sense of Kato, 
and the essential spectrum of $\widetilde{P}_\theta$ is given by 
\[
\sigma_{\mathrm{ess}}(\widetilde{P}_{\theta})
=\bigl\{(\xi+\theta\sin(\pi \xi))^2\bigm| \hspace{0.1cm}\xi \in \mathbb{R}\bigr\}.
\]
We note the critical values of $(\xi+\theta\sin(\pi \xi))^2$ is $\mathcal{T}=\{n^2\mid n\in \mathbb{N}\cup\{0\}\}$. 
We fix $n\in\mathbb{N}$, and for the energy interval $((n-1)^2, n^2)$, we take $\theta=(-1)^n i \delta=\pm i \delta$. 
We easily see that for $0<\delta<\pi^{-1}$ 
the essential spectrum of $\widetilde{P}_{\pm i \delta}$ is the graph of a function 
and we define $\kappa_{\pm \delta}(x)$, $x\ge 0$, by the relation 
\[
\sigma_{\mathrm{ess}}(\widetilde{P}_{\pm i \delta})=\bigl\{z=x+iy  \bigm|  y=\kappa_{\pm \delta}(x), x\ge 0\bigr\}.
\]
We choose $0<\delta<\delta_0=\min \{\pi^{-1}, K\}$, and we 
set 
\[
\Omega_{n, \delta}=\bigl\{z=x+iy \bigm| (n-1)^2<x<n^2, y>\kappa_{(-1)^n \delta}(x)\bigr\}.
\]

\begin{proof}[Proof of Theorem~\ref{thm-1}]
We fix $n\in\mathbb{N}$ and $\delta>0$ as above, and we denote 
$\mathcal{A}=L_{\mathrm{comp}}^{2}(\mathbb{R})$. 
We first note that $U_{\theta}\hat{f}$ ($f\in \mathcal{A}$) has an analytic 
continuation for complex $\theta$. 
We denote the resolvent $R_{+}(z)$ on the Fourier space by $\widetilde{R_{+}}(z)$.
For  $f, g \in \mathcal{A}$, we have 
\[(\hat{f}, \widetilde{R_{+}}(z)\hat{g})
= (U_{\theta}\hat{f}, U_{\theta}\widetilde{R_{+}}(z)U_{\theta}^{-1}U_{\theta}\hat{g})
=(U_{\overline{\theta}}\hat{f}, (z-\widetilde{P_{\theta}})^{-1}U_{\theta}\hat{g})
\]
where $\theta\in\mathbb{R}$ and $\mathrm{Im}z>0$. 
The right hand side is analytic with respect to $\theta$ by Lemma~\ref{lem-1}, 
where $\theta$ ranges over a complex neighborhood of $\{(-1)^n i \delta \mid 0\le \delta <\delta_0\}$. 
This in turn implies that the left hand side has a meromorphic continuation to $\Omega_{n, \delta_0}$ 
with respect to $z$. 
Thus Theorem~\ref{thm-1} is proved for $\Omega=
\bigcup_{n\in \mathbb{N}}\Omega_{n, \delta_0}$.
\end{proof}

\begin{remark}\label{rem-2-2}
We set $\Pi_z^{\theta}=\frac{1}{2\pi i}\oint_{C(z)}(\zeta-\widetilde{P_{\theta}})^{-1}d\zeta$ 
be the spectral projection for $\widetilde{P_{\theta}}$. 
Then we have 
\[
\frac{1}{2\pi i}\oint_{C(z)}(f, R_{+}(\zeta)g)d\zeta=\frac{1}{2\pi i}\oint_{C(z)} 
(U_{\overline{\theta}}\hat{f}, (\zeta-\widetilde{P_{\theta}})^{-1}U_{\theta}\hat{g})d\zeta
=(U_{\overline{\theta}}\hat{f}, \Pi_z^{\theta}U_{\theta}\hat{g}).
\]
We note that $\{U_{\theta}\hat{f}\mid f\in \mathcal{A}\}$ is dense in $L^2$, 
which is proved by an argument similar to \cite[Theorem 3]{H}.
This implies that $m_z=\mathrm{rank}[\Pi_z^{\theta}]$. 
Namely, the resonances coincide with the discrete eigenvalues of $\widetilde{P_{\theta}}$ 
including multiplicities.  
In particular, $\mathrm{Res}(P)$ is discrete and $m_z<\infty$ for any $z\in \Omega$. 
\end{remark}

\begin{remark}
In the case of $V=a\frac{\sin 2x}{x}+V_0$, $V_0 \in C^{\infty}_c(\mathbb{R}; \mathbb{R})$, 
Remark~\ref{rem-2-1} and the proof of Lemma~\ref{lem-1} show that Lemma~\ref{lem-1} holds for 
$\theta \in \mathbb{C}\setminus (-\pi^{-1}, \pi^{-1})$. 
Thus the set of resonances $\mathrm{Res}_n(P)$ is defined in $\mathbb{C}\setminus (0, \infty)$ 
for any $n\in \mathbb{N}$ including multiplicities 
by the meromorphic continuation of $(f, R_+(z)g)$ from $\{z\mid  0<\arg z<\pi\}$ to 
\[
\{z\mid 0<\arg z<\pi\}\cup \{z\mid\arg z=0, (n-1)^2<|z|<n^2\}\cup \{z\mid-2\pi<\arg z<0\}.
\]
This poses the problem of whether $\mathrm{Res}_n(P)\not=\mathrm{Res}_{n'}(P)$ when $n\not=n'$.
\end{remark}

\section{Viscosity limit}
For the notational simplicity, we set $P_0=P$, $\widetilde{P}_0=\widetilde{P}$ 
and $\widetilde{P}_{0,\theta}=\widetilde{P}_{\theta}$.
In the Fourier space, $P_{\varepsilon}$, $\varepsilon \ge 0$, has the following form: 
\[
\widetilde{P}_{\varepsilon}=\xi^2+\widetilde{V}+i\varepsilon \partial_{\xi}^2.
\]
Hence the distorted operator $\widetilde{P}_{\varepsilon, \theta}
=U_{\theta}\widetilde{P}_{\varepsilon} U_{\theta}^{-1}$
is given by 
\[
\widetilde{P}_{\varepsilon, \theta}
=(\xi+\theta\sin(\pi \xi))^2+\widetilde{V}_{\theta}-i\varepsilon D_{\xi}(1+\pi \theta \cos(\pi \xi))^{-2}D_{\xi}
-i\varepsilon r_{\theta}(\xi),
\]
where $r_{\theta}(\xi)=-\Phi_{\theta}'(\xi)^{-\frac{1}{2}}\partial_{\xi}
(\Phi_{\theta}'(\xi)^{-1}\partial_{\xi}(\Phi_{\theta}'(\xi)^{-\frac{1}{2}}))$ 
is a function which is analytic with respect to $\theta$ and bounded with respect to $\xi$. 
Since $\widetilde{P}_{\varepsilon,\theta}$ is elliptic, 
$\widetilde{P}_{\varepsilon, \theta}$, $\varepsilon>0$, has purely discrete spectrum.
Moreover, for fixed $\varepsilon>0$, 
$\widetilde{P}_{\varepsilon,\theta}$ is analytic with respect to $\theta$ in the sense of Kato. 
These imply that the eigenvalues of 
$\widetilde{P}_{\varepsilon, \theta}$ coincide with those of $\widetilde{P}_{\varepsilon}$ 
including multiplicities by the same argument as in Remark~2.2.
Thus it is enough to show that the eigenvalues of $\widetilde{P}_{\varepsilon, \theta}$ 
converge to those of $\widetilde{P}_{\theta}$ as $\varepsilon \to +0$.

\begin{proof}[Proof of Theorem~\ref{thm-2}]
We first prove the resolvent estimate for the following distorted free Hamiltonian
\[
\widetilde{Q}_{\varepsilon, \theta}=(\xi+\theta\sin(\pi \xi))^2-i\varepsilon D_{\xi}
(1+\pi \theta \cos(\pi \xi))^{-2}D_{\xi}-i\varepsilon r_{\theta}(\xi), 
\quad \varepsilon \ge 0. 
\]
In the following, we fix $n\in \mathbb{N}$, 
set $\theta=(-1)^n i \delta=\pm i \delta$, $0<\delta<\delta_0$  as in Section 2.

We set $h=\sqrt{\varepsilon}$ and view $\widetilde{Q}_{\varepsilon, \theta}$ 
as an $h$-pseudodifferential operator in the Fourier space. 
We easily see that the numerical range of the $h$-principal symbol of $\widetilde{Q}_{\varepsilon, \theta}$, i.e., 
\[
\bigl\{ (\xi+\theta\sin(\pi \xi))^2-i(1+\pi \theta \cos(\pi \xi))^{-2}x^2\bigm| x,\xi\in\mathbb{R}\bigr\}
\] 
is disjoint form the set 
\[
\Omega'_{n, \delta}=\Omega_{n, \delta}\cap\biggl\{z=x+iy\biggm| 
y>\frac{1}{2}\Bigl(\frac{1}{\pi \delta}-\pi \delta\Bigr)\bigl((n-1)^2-x\bigr), 
y>\frac{1}{2}\Bigl(\frac{1}{\pi \delta}-\pi \delta\Bigr)\bigl(x-n^2\bigr) \biggr\}.
\]
We note that a simple computation shows that $\Omega'_{n, \delta}=\Omega_{n, \delta}$ for $0<\delta<\frac{1}{\sqrt{3}\pi}$. 
Now we fix $z \in \Omega'_{n, \delta}$. 
Then there exists $\rho_0>0$ such that there is no resonance in 
$B(z, \rho_0)\Subset \Omega'_{n, \delta}$ possibly expect for $z$, 
where $B(z,\rho)$ denotes the ball of radius $\rho$ with the center at $z$. 
In the following, we fix $0<\rho<\rho_0$, and let $w\in B_z=B(z, \rho)$.
By the standard semiclassical calculus we learn $(\widetilde{Q}_{\varepsilon, \theta}-w)^{-1}$ exists 
and  $\|(\widetilde{Q}_{\varepsilon, \theta}-w)^{-1}\|_{L^2 \to L^2}\le C$ for 
$w \in B_z$ and for sufficiently  small $\varepsilon>0$. 
We note that it also holds for $\varepsilon=0$.

We next employ the perturbation argument. 
Since $(\widetilde{Q}_{\varepsilon, \theta}-w)^{-1}$ exists, we have 
\[ 
\widetilde{P}_{\varepsilon, \theta}-w
=(1+\widetilde{V}_{\theta}(\widetilde{Q}_{\varepsilon, \theta}-w)^{-1})(\widetilde{Q}_{\varepsilon, \theta}-w).
\]
By Lemma~\ref{lem-1} and the boundedness of $(\xi^2+i)(\widetilde{Q}_{\varepsilon, \theta}-w)^{-1}$,
we learn $\widetilde{V}_{\theta}(\widetilde{Q}_{\varepsilon, \theta}-w)^{-1}$ is compact for $\varepsilon \ge 0$. 
Thus the analytic Fredholm theory can be applied, and 
the number of the eigenvalues of $P_{\varepsilon, \theta}$, $\varepsilon \ge 0$, in $B_z$ is given by  
\[  
\mathrm{tr}\oint_{\partial B_z} (w-\widetilde{P}_{\varepsilon, \theta})^{-1}dw
=\mathrm{tr}\oint_{\partial B_z} (\partial_w\widetilde{V}_{\theta}(\widetilde{Q}_{\varepsilon, \theta}-w)^{-1})
(1+\widetilde{V}_{\theta}(\widetilde{Q}_{\varepsilon, \theta}-w)^{-1})^{-1}dw.
\]
Thus operator-valued Rouch\'{e}'s theorem (\cite[Theorem~2.2]{GS}) implies that in order to prove Theorem~\ref{thm-2}, 
it suffices to show 
\[
\bigl\| ((1+\widetilde{V}_{\theta}(\widetilde{Q}_{0, \theta}-w)^{-1})
-(1+\widetilde{V}_{\theta}(\widetilde{Q}_{\varepsilon, \theta}-w)^{-1}))
(1+\widetilde{V}_{\theta}(\widetilde{Q}_{0, \theta}-w)^{-1})^{-1}\bigr\|_{L^2 \to L^2}<1
\] 
for $w \in \partial B_z$ and small $\varepsilon>0$. 
Since $(1+\widetilde{V}_{\theta}(\widetilde{Q}_{0, \theta}-w)^{-1})^{-1}$ exists and 
independent of $\varepsilon>0$ for $w \in \partial B_z$, the above estimate hold if we show 
\begin{equation}\label{eq-3-1}
\lim_{\varepsilon \to 0}\| \widetilde{V}_{\theta}(\widetilde{Q}_{0, \theta}-w)^{-1}
-\widetilde{V}_{\theta}(\widetilde{Q}_{\varepsilon, \theta}-w)^{-1}\|_{L^2 \to L^2}=0
\end{equation}
uniformly for $w\in \partial B_z$. 

Let $\gamma>0$. 
By the arguments in the proof of Lemma~\ref{lem-1}, 
we can decompose $\widetilde{V}_{\theta}=\widetilde{V}_{\theta, 1}+\widetilde{V}_{\theta, 2}$, 
where $\widetilde{V}_{\theta, 1}$ is a smoothing pseudodifferential 
operator in the Fourier space and $\|\widetilde{V}_{\theta, 2}\|_{L^2 \to L^2}<\gamma$.
Since $\|(\widetilde{Q}_{\varepsilon, \theta}-w)^{-1}\|_{L^2 \to L^2}\le C$ 
for small $\varepsilon\ge 0$ and $w\in B_z$, we have 
\[
\|\widetilde{V}_{\theta, 2}(\widetilde{Q}_{0, \theta}-w)^{-1}
-\widetilde{V}_{\theta, 2}(\widetilde{Q}_{\varepsilon, \theta}-w)^{-1}\|_{L^2 \to L^2}\le 2C\gamma,
\] 
where $C$ is independent of $\gamma$. 
By the resolvent equation, we also learn 
\begin{align*}
&\widetilde{V}_{\theta, 1}(\widetilde{Q}_{0, \theta}-w)^{-1}
-\widetilde{V}_{\theta, 1}(\widetilde{Q}_{\varepsilon, \theta}-w)^{-1}\\
&=-i \varepsilon \widetilde{V}_{\theta, 1}(\widetilde{Q}_{0, \theta}-w)^{-1}
(D_{\xi}(1+\pi \theta \cos(\pi \xi))^{-2}D_{\xi}+ r_{\theta}(\xi))(\widetilde{Q}_{\varepsilon, \theta}-w)^{-1}.
\end{align*}
Since $\widetilde{V}_{\theta, 1}$ is a smoothing pseudodifferential 
operator and 
$(\widetilde{Q}_{0, \theta}-w)^{-1}$ is also a pseudodifferential operator with a bounded symbol, 
$\widetilde{V}_{\theta, 1}(\widetilde{Q}_{0, \theta}-w)^{-1}
D_{\xi}^2$ is $L^2$-bounded. 
Thus we have 
\[
\|\widetilde{V}_{\theta, 1}(\widetilde{Q}_{0, \theta}-w)^{-1}
-\widetilde{V}_{\theta, 1}(\widetilde{Q}_{\varepsilon, \theta}-w)^{-1}\|_{L^2 \to L^2}
\le C_{\gamma} \varepsilon
\]
with some ($\gamma$-dependent) constant $C_{\gamma}>0$. 
If $\varepsilon$ is so small that $\varepsilon\leq (C/C_{\gamma})\gamma$, we have 
\[
\| \widetilde{V}_{\theta}(\widetilde{Q}_{0, \theta}-w)^{-1}
-\widetilde{V}_{\theta}(\widetilde{Q}_{\varepsilon, \theta}-w)^{-1}\|_{L^2 \to L^2}
\le 2C\gamma+C_{\gamma} \varepsilon\le 3C\gamma
\]
and thus \eqref{eq-3-1} is proved since $\gamma>0$ may be arbitrary small. 
Since $0<\delta<\delta_0$ is arbitrary, Theorem~\ref{thm-2} is proved for 
$\Omega=\bigcup_{n\in \mathbb{N}}\Omega'_{n, \delta_0}$.
\end{proof}

\begin{remark}
In the case of $V(x)=a\frac{\sin 2x}{x}$, the above proof is simpler. Namely, 
it suffices to use the decomposition: 
\begin{align*}
\widetilde{V}_{\theta}&=(\Phi_{\theta}')^{\frac{1}{2}}\widetilde{V}(\Phi_{\theta}')^{\frac{1}{2}}
=(\Phi_{\theta}')^{\frac{1}{2}} \widetilde{V}_{1, R} (\Phi_{\theta}')^{\frac{1}{2}}
+(\Phi_{\theta}')^{\frac{1}{2}}\widetilde{V}_{2, R}(\Phi_{\theta}')^{\frac{1}{2}}
=\widetilde{V}_{\theta, 1}+\widetilde{V}_{\theta, 2}
\end{align*}
for large $R>0$, where $\widetilde{V}_{j, R}$ is the Fourier multiplier 
on the Fourier space by $V_{j, R}$, $\chi \in C_c^{\infty}(\mathbb{R})$ such that $\chi=1$ near $x=0$, and
$a\frac{\sin 2x}{x}=a\frac{\sin 2x}{x}\chi(x/R)+a\frac{\sin 2x}{x}(1-\chi(x/R))=V_{1, R}+V_{2, R}$.
\end{remark}


\begin{thebibliography}{99}
\bibitem{B1}
H. Behncke, Absolute Continuity of Hamiltonians with von Neumann Wigner Potentials, 
Proc. Amer. Math. Soc. 111, (1991), 373-384.
\bibitem{B2}
H. Behncke, The $m$-function for Hamiltonians with Wigner-von Neumann potentials, 
J. Math. Phys. 35, (1994), 1445-1462. 
\bibitem{CHM}
J. Cruz-Sampedro, I. Herbst and R. Martinez-Avenda\~{n}o, Perturbations of the Wigner-Von Neumann 
Potential Leaving the Embedded Eigenvalue Fixed, Ann. Henri Poincar\'{e}, 3, (2002), 331-345.
\bibitem{C}
H. L. Cycon, Resonance defined by modified dilations, Helv. Phys. Acta. 58, (1989), 969-987. 
\bibitem{DR}
N. Dang and G. Riviere, Pollicott-Ruelle spectrum and Witten Laplacians, 
arXiv: 1709.04265, to appear in J. Eur. Math. 
\bibitem{DMR}
A. Devinatz, R. Moeckel and P. Rejto, A Limiting Absorption Principle for Schr\"{o}dinger 
Operators with von Neumann-Wigner Type Potentials, Integr. Equat. Oper. Th. 14, (1991), 13-68.
\bibitem{Dr}
A. Drout, Pollicott-Ruelle resonances via kinetic Brownian motion, 
Commun. Math. Phys. 356, (2017), 357-396.
\bibitem{DZ}
S. Dyatlov and M. Zworski, Stochastic stability of Pollicott-Ruelle resonances, 
Nonlinearity. 28, (2015), 3511-3533.  
\bibitem{FH}
R. Froese and I. Herbst, Exponential Bounds and Absence of Positive Eigenvalues 
for N-body Schr\"{o}dinger Operators, Commun. Math. Phys. 87, (1982), 429-447.
\bibitem{GS}
I. C. Gohberg and E. I. Sigal, An Operator generalization of the logarithmetic residue theorem 
and the theorem of Rouch\'{e}, Mt. Sb. (N.S.). 84, (1971), 607-629.
\bibitem{GZ}
J. Galkowski and M. Zworski, Viscosity limits for 0th order pseudodifferential operators, 
arXiv: 1912.09840. 
\bibitem{HKS}
D. B. Hinton, M. Klaus and J. K. Shaw, Embedded Half-bound States for Potentials of Wigner-von Neumann Type, 
Proc. London Math. Soc. (3) 62, (1991), 607-646.  
\bibitem{H}
W. Hunziker, Distortion analyticity and molecular resonance curves, 
Ann. Inst. Henri Poincar\'{e}. 45, (1986), 339-358.
\bibitem{K}
M. Klaus, Asymptotic behavior of Jost functions near resonance points for Wigner-von Neumann 
type potentials, J. Math. Phys. 32, (1991), 163-172.
\bibitem{L}
M. Lukic, Schr\"{o}dinger operators with slowly decaying Wigner-von Neumann type potentials, 
J. Spectr. Theory. 3, (2013), 147-169. 
\bibitem{N}
S. Nakamura, Distortion analyticity for two-body Schr\"{o}dinger operators, 
Ann. Inst. Henri Poincar\'{e}. 53, (1990), 144-157.  
\bibitem{N2}
S. Nakamura, A remark on the Mourre theory for two body Schr\"{o}dinger operators, 
J. Spectral Theory. 4, (2015), 613-619.
\bibitem{NW}
J. von Neumann and E. Wigner, Uber merkw\"{u}rdige diskrete Eigenwerte, Phys. Z. 30, 
(1929), 465-467. 
\bibitem{NZ1}
S. Nonnenmacher and M. Zworski, Quantum decay rates in chaotic scattering, 
Acta. Math. 203, (2009), 149-233. 
\bibitem{NZ2}
S. Nonnenmacher and M. Zworski, Decay of correlations for normally hyperbolic trapping, 
Invent. Math. 200, (2015), 345-438.
\bibitem{RS}
M. Reed and B. Simon, Methods of Modern Mathematical Physics IV, 
Academic Press, 1978.
\bibitem{R}
P. Rejto and M. Taboada, A Limiting Absorption Principle for Schr\"{o}dinger Operators with Generalized Von Neumann-Wigner 
Potentials I, J. Math. Anal. Appl. 208, (1997), 85-108. 
\bibitem{RUU}
S. Richard, J. Uchiyama and T. Umeda, Schr\"{o}dinger operaotrs with $n$ positive eigenvalues: 
an explicit construction involving complex-valued potentials, Proc. Japan Acad. 92, Ser. A, (2016), 7-12. 
\bibitem{S}
I. M. Sigal, Complex transformation method and resonances in one-body quantum systems, 
Ann. Inst. Henri Poincar\'{e}. 41, (1984), 103-114. 
\bibitem{St}
P. Stefanov, Approximating Resonances with the Complex Absorbing Potential Method, 
Commun. Part. Diff. Eq. 30, (2005), 1843-1862.
\bibitem{V}
A. Vasy, Microlocal analysis of asymptotically hyperbolic and Kerr-de Sitter spaces, with an 
appendix by Semyon Dyatlov, Invent. Math. 194, (2013), 381-513.
\bibitem{Z}
M. Zworski, Scattering resonances as viscosity limits, in Algebraic and Analytic Microlocal Analysis, 
M.Hitrik, D. Tamarkin, B. Tsygan, and S. Zeldithch, eds. Springer 2018.
\bibitem{Z2}
M. Zworski; Semiclassical Analysis, 
American Math. Soc., 2012.
\end{thebibliography}
\end{document}